\documentclass[conference]{IEEEtran}
\IEEEoverridecommandlockouts

\usepackage{cite}
\usepackage{amsmath,amssymb,amsfonts}

\usepackage{graphicx}
\usepackage{xcolor}
\usepackage{algorithm}
\usepackage{algpseudocode}
\usepackage{float}
\usepackage{booktabs}
\usepackage{multirow}
\usepackage{braket}
\usepackage{url}
\usepackage{makecell}
\usepackage{subcaption}
\usepackage[T1]{fontenc}
\usepackage[utf8]{inputenc}

\usepackage{tikz}
\usetikzlibrary{positioning,shapes.geometric,arrows.meta,decorations.pathreplacing}
\usepackage{quantikz}

\newcommand{\Ulabel}{\ensuremath{\mathrm{U}}}

\usepackage[hidelinks]{hyperref}

\usepackage{amsthm}
\numberwithin{equation}{section}
\newtheorem{assumption}{Assumption}[section]

\newtheorem{lemma}{Lemma}[section]
\newtheorem{proposition}{Proposition}[section]
\newtheorem{corollary}{Corollary}[section]
\newtheorem{theorem}{Theorem}[section]

\theoremstyle{remark}
\newtheorem{remark}{Remark}[section]

\def\BibTeX{{\rm B\kern-.05em{\sc i\kern-.025em b}\kern-.08em
    T\kern-.1667em\lower.7ex\hbox{E}\kern-.125emX}}
\begin{document}

\title{qc-kmeans: A Quantum Compressive K-Means Algorithm for NISQ Devices}


\author{%
\IEEEauthorblockN{Pedro Chumpitaz-Flores\textsuperscript{*}}
\IEEEauthorblockA{%
\textit{University of South Florida}\\
Tampa, FL, USA\\
pedrochumpitazflores@usf.edu}
\and
\IEEEauthorblockN{My Duong\textsuperscript{*}}
\IEEEauthorblockA{%
\textit{University of South Florida}\\
Tampa, FL, USA\\
myduong@usf.edu}
\and
\IEEEauthorblockN{Ying Mao}
\IEEEauthorblockA{%
\textit{Fordham University}\\
New York, NY, USA\\
ymao41@fordham.edu}
\and
\IEEEauthorblockN{Kaixun Hua}
\IEEEauthorblockA{%
\textit{University of South Florida}\\
Tampa, FL, USA\\
khua@usf.edu}
\thanks{\textsuperscript{*}Equal  Contribution.}%
}

\maketitle

\begin{abstract}
Clustering on NISQ hardware is constrained by data loading and limited qubits. We present \textbf{qc-kmeans}, a hybrid compressive $k$-means that summarizes a dataset with a constant-size Fourier-feature sketch and selects centroids by solving small per-group QUBOs with shallow QAOA circuits. The QFF sketch estimator is unbiased with mean-squared error $O(\varepsilon^2)$ for $B,S=\Theta(\varepsilon^{-2})$, and the peak-qubit requirement
$q_{\text{peak}}=\max\{D,\lceil \log_2 B\rceil + 1\}$
does not scale with the number of samples. A refinement step with elitist retention ensures non-increasing surrogate cost.
In Qiskit Aer simulations (depth $p{=}1$), the method ran with $\le 9$ qubits on low-dimensional synthetic benchmarks and achieved competitive sum-of-squared errors relative to quantum baselines; runtimes are not directly comparable. On nine real datasets (up to $4.3\times 10^5$ points), the pipeline maintained constant peak-qubit usage in simulation. Under IBM noise models, accuracy was similar to the idealized setting. Overall, qc-kmeans offers a NISQ-oriented formulation with shallow, bounded-width circuits and competitive clustering quality in simulation.
\end{abstract}


\section{Introduction} \label{sec:introduction}
Unsupervised $k$-means is a standard method for partitioning datasets into $k$ cohesive groups, widely used in data mining, knowledge discovery, and pattern recognition \cite{rao_cluster_1971, hartigan_algorithm_1979, jain_data_2010}. In its standard Lloyd iteration, each step computes distances from all $N$ points to $k$ centroids in $d$ dimensions, yielding a per-iteration cost of $O(kNd)$ \cite{lloyd_least_1982,jain_data_2010}. While many improvements have been proposed,
clustering large, high-dimensional, and multi-cluster datasets remains computationally demanding, motivating the search for new algorithmic paradigms.

Quantum computing has attracted interest for addressing clustering challenges \cite{saiphet2021quantum,chen2025provably}. Pure and hybrid approaches leverage quantum superposition and amplitude-estimation ideas to parallelize key subroutines and accelerate linear-algebraic steps such as eigenvalue estimation \cite{brassard2000quantum,lloyd2014qpca,kerenidis2021qspectral}. However, several proposals assume advanced data access (e.g., QPU, QRAM) and deep circuits suited to fault-tolerant regimes \cite{kerenidis2019q,chen2025provably,giovannetti2008qram,preskill2018quantum}. In practice, QRAM is unrealized on current hardware \cite{weiss2024quantum}, and today’s Noisy Intermediate-Scale Quantum (NISQ) devices have limited qubits and suffer decoherence and gate errors \cite{preskill2018quantum}. Moreover, while amplitude estimation can offer quadratic speedups for mean or distance estimation, standard variants require circuit depths that strain NISQ hardware, motivating low-depth or iterative alternatives \cite{brassard2000quantum,giurgica2022lowdepth,suzuki2020qae}. These constraints have spurred interest in hybrid approaches that pair small quantum processors with classical data reduction to keep the quantum footprint small \cite{harrow_small_2020,tomesh_coreset_2020,yogendran2024big}. 

\begin{figure}[htbp]
    \centering
    \includegraphics[width=\linewidth]{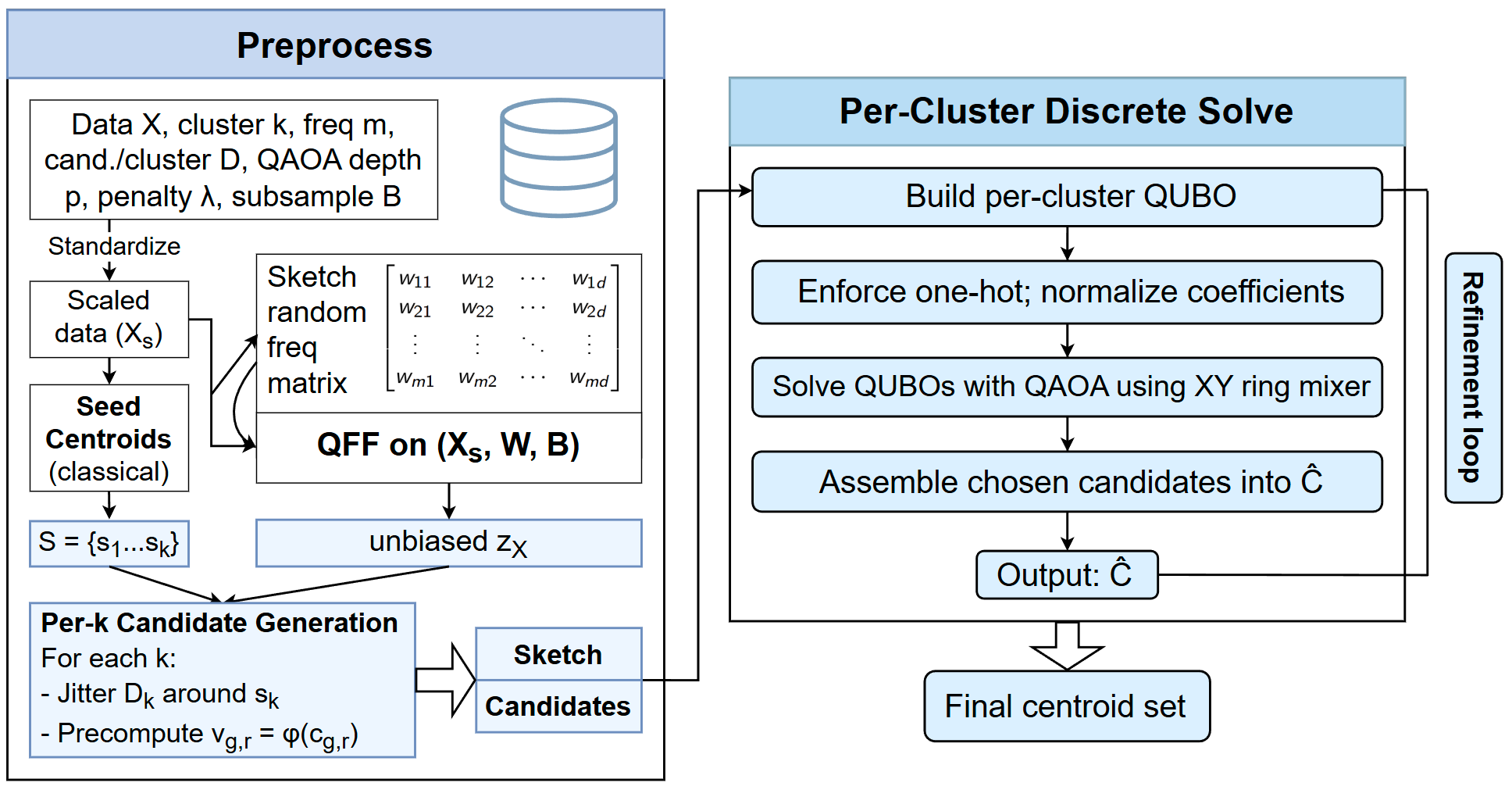}
    \caption{Qc-kmeans formulation}
    \label{fig:qckmeans_form}
\end{figure}

\textbf{Our Contributions:} We propose a hybrid \emph{Quantum Compressive $k$-Means (qc-kmeans)} algorithm for scalable quantum clustering on NISQ hardware. Qc-kmeans compresses the dataset into a compact sketch and runs quantum subroutines only on this reduced representation, trading approximation for practicality and reducing memory and qubit requirements. The method isolates a small quantum subproblem; \emph{comparisons of SSE and runtime are restricted to quantum baselines and small-QUBO references}, with classical baselines included only for context and without claims of superiority.

\section{Related Works} \label{sec:related_works}

$K$-means clustering has been an early target for quantum algorithmic design. \cite{kerenidis2019q} first developed the q-means algorithm, which mirrors the classical $k$-means but uses quantum linear-algebra subroutines to accelerate centroid updates and distance computations. The q-means algorithm encodes data points in quantum amplitude vectors, typically assuming they reside in QRAM for efficient access. Under these conditions, q-means can find approximate cluster centroids with high probability. Q-means and its improved variants \cite{doriguello2023you, chen2025provably} run in polylogarithmic time per iteration, providing an exponential speedup over the $O(kNd)$ classical cost. As discussed in §\ref{sec:introduction}, these speedups rely on assumptions such as QRAM and deeper circuits that are challenging for current NISQ devices.


To bridge the gap between algorithmic theory and today’s NISQ hardware, recent work has focused on hybrid quantum–classical clustering. \cite{fuchs2021efficient} used QAOA to obtain approximate solutions for weighted max-$k$-cut. Building on distance estimation via small SWAP-test circuits, \cite{DiAdamo2022Practical} analyzed noise pitfalls on NISQ devices and proposed a parallelized quantum $k$-means that reduces circuit depth by estimating multiple distances concurrently, improving performance on an energy-grid dataset. In a related vein, \cite{poggiali2024hybrid} presented hybrid quantum clustering that computes distances from a point to all $k$ centroids in superposition, yielding runtime gains for sufficiently large datasets without degrading clustering quality.

Orthogonal to quantum progress, classical ML tackles scale via summarization. \emph{Classical compressive $k$-means (CKM)} compresses a dataset into a fixed-size random-feature sketch, making clustering cost independent of $N$ \cite{Keriven2017Compressive}. Hybrid approaches build classical coresets to enable quantum clustering but currently handle only 2- or 3-means \cite{harrow_small_2020, tomesh_coreset_2020, yogendran2024big}. A quantum coreset builder runs in $\tilde O(\sqrt{N k},d^{3/2})$ time, yielding a quadratic speedup for large $N$ \cite{xue2023Near}. Despite these advances, scalability remains limited: fully quantum algorithms promise asymptotic gains but are not yet deployable \cite{yogendran2024big}, hybrids have shown only small-scale results (hundreds of points), and compressive methods are rarely combined with quantum speedups. This work addresses the gap with \emph{qc-kmeans}, integrating quantum algorithms, hybrid parallelism, and compressive learning to push NISQ-era clustering.


\section{Quantum Compressive $k$-Means via Fourier Features} \label{sec:ckm_problem}

\subsection{Classical $k$-Means}
Given a dataset \(X=\{x_1,\ldots,x_N\}\subset\mathbb{R}^d\) with \(N\) samples and \(d\) dimensions, the $k$-Means task assigns points to \(k\) clusters with centroids \(\{\mu_j\}_{j=1}^k\) that minimize the within-cluster sum of squares:
\(
\mathrm{WCSS} \;=\; \sum_{j=1}^{k} \sum_{x_i \in C_j} \bigl\|x_i - \mu_j\bigr\|_2^2 ,
\)
where \(\mu_j\) is the mean of points assigned to cluster \(j\). Lloyd’s heuristic \cite{lloyd_least_1982} has per-iteration cost \(O(kNd)\), which becomes prohibitively expensive as the dataset grows in size and dimensionality. To address this scalability bottleneck, random projections can be applied to lower computational cost while approximately preserving clustering structure.
 
\subsection{Quantum compressive $k$-Means}
Given \(X=\{x_1,\ldots,x_N\}\subset\mathbb{R}^d\), we construct a matrix of random frequencies as
\[
W
=
\begin{bmatrix}
w_{11} & w_{12} & \cdots & w_{1d} \\
w_{21} & w_{22} & \cdots & w_{2d} \\
\vdots & \vdots & \ddots & \vdots \\
w_{m1} & w_{m2} & \cdots & w_{md}
\end{bmatrix}
\in\mathbb{R}^{m\times d}.
\]
Each row vector \(w_j = (w_{j1},\ldots,w_{jd})^\top\) for \(j\in\{1,\ldots,m\}\) is independently sampled from the multivariate normal distribution \(\mathcal{N}(0,\sigma^2 I_d)\).
We then define the complex Fourier feature map
\(
\phi(x) \;=\; \exp\big(\mathbf{i}\,W x\big) \in \mathbb{C}^m,
\)
where the exponential is applied elementwise for all \(x\in\mathbb{R}^d\) and where \(\mathbf{i}^2=-1\). 
The dataset is summarized by the compressed feature mean:
\begin{equation}\label{eqn:features_data}
\begin{aligned}
z_X &= \frac{1}{N}\sum_{i=1}^N \phi(x_i)\in\mathbb{C}^m, \\
x_i &\in \mathbb{R}^d \quad \forall\, i\in\{1,\ldots,N\}.
\end{aligned}
\end{equation}

\noindent For a centroid set \(C=\{\mu_1,\ldots,\mu_k\}\subset\mathbb{R}^d\), its feature mean is
\begin{equation}\label{eqn:features_centroids}
\begin{aligned}
z_\mu &= \frac{1}{k}\sum_{g=1}^k \phi(\mu_g)\in\mathbb{C}^m, \\
\mu_g &\in \mathbb{R}^d \quad \forall\, g\in\{1,\ldots,k\}.
\end{aligned}
\end{equation}

\noindent Compressive \(k\)-means (CKM) \cite{Keriven2017Compressive} aligns \(z_\mu\) to \(z_X\) by minimizing \(\|z_X-z_\mu\|_2^2\).
\begingroup
\setlength{\jot}{1pt}
\begin{subequations}\label{eqn:ckm_base}
\begin{alignat}{3}
&\min_{\mu_1,\ldots,\mu_k \in \mathbb{R}^d} && \bigl\|\,z_X-\tfrac{1}{k}\sum_{g=1}^k \phi(\mu_g)\,\bigr\|_2^2 \\
&\text{subject to} \quad && \mu_g \in \mathbb{R}^d, && \forall\, g\in\{1,\ldots,k\}.
\end{alignat}
\end{subequations}
\endgroup

\noindent To obtain a quantum optimization model, we discretize the centroid search. For each cluster index \(g\in\mathcal{K}=\{1,\ldots,k\}\) let \(\mathcal{C}_g=\{c_{g,1},\ldots,c_{g,D_g}\}\subset\mathbb{R}^d\) be a finite candidate set, where each candidate \(c_{g,r}\in\mathbb{R}^d\) for \(r\in\{1,\ldots,D_g\}\), and introduce binary decision variables \(y_{g,r}\in\{0,1\}\) indicating the selection of \(c_{g,r}\). One-hot constraints enforce exactly one candidate per cluster:
\begin{subequations}\label{eqn:onehot_centroid}
\begin{align}
\sum_{r=1}^{D_g} y_{g,r} &= 1, && \forall\, g\in\mathcal{K},\\
y_{g,r} &\in \{0,1\}, && \forall\, g\in\mathcal{K},\ \forall\, r\in\{1,\ldots,D_g\}.
\end{align}
\end{subequations}
Let \(v_{g,r}=\varphi(c_{g,r})=\exp(\mathrm{i}\,W c_{g,r})\in\mathbb{C}^m\) for all \(g\in\mathcal{K}\) and \(r\in\{1,\ldots,D_g\}\). Then
\begin{equation}\label{eqn:zmu_linear}
z_\mu(y) \;=\; \frac{1}{k}\sum_{g=1}^k\sum_{r=1}^{D_g} y_{g,r}\, v_{g,r}.
\end{equation}

\noindent Expanding \(\|z_X-z_\mu(y)\|_2^2\) and dropping the constant \(\|z_X\|_2^2\) yields a real quadratic polynomial in \(\{y_{g,r}\}\). With the complex (Hermitian) inner product \(\langle a,b\rangle:=\sum_{j=1}^m \overline{a_j}\,b_j\), define
\begin{equation}\label{eqn:coeffs}
\begin{aligned}
b_{g,r} &= -\frac{2}{k}\,\mathrm{Re}\,\langle z_X, v_{g,r}\rangle, \\
Q_{(g,r),(g',r')} &= \frac{1}{k^2}\,\mathrm{Re}\,\langle v_{g,r}, v_{g',r'}\rangle .
\end{aligned}
\end{equation}

\noindent Embedding the one-hot constraints as penalties with a parameter \(\lambda>0\) leads to the QUBO Hamiltonian
\begin{subequations}\label{eqn:qubo_ckm}
\begin{align}
\begin{split}
H_{\mathrm{fit}}(y)
&= \sum_{g,r} b_{g,r}\, y_{g,r} \\
&\quad+\sum_{(g,r)}\sum_{(g',r')} Q_{(g,r),(g',r')}\, y_{g,r}\, y_{g',r'}
\end{split}\\
H_{\mathrm{one\text{-}hot}}(y)
&= \lambda \sum_{g\in\mathcal{K}}\Bigl(1-\sum_{r=1}^{D_g} y_{g,r}\Bigr)^2,\\
H_{\mathrm{CKM}}(y) &= H_{\mathrm{fit}}(y) + H_{\mathrm{one\text{-}hot}}(y).
\end{align}
\end{subequations}

A sufficient choice ensuring any infeasible (non one-hot) assignment has higher energy than all feasible ones is
\begin{equation}\label{eqn:lambda}
\lambda \;>\; \sum_{g,r} |b_{g,r}| \;+\; \sum_{(g,r)\ne(g',r')} |Q_{(g,r),(g',r'})| \;+\; \varepsilon,
\end{equation}
for any fixed \(\varepsilon>0\).
To make this bound practical and avoid overly conservative values of \(\lambda\), the coefficients are normalized before the penalty is set. Specifically, let
\[
S_{\text{coef}} \;:=\; \sum_{g,r} |b_{g,r}| \;+\; \sum_{(g,r),(g',r')} |Q_{(g,r),(g',r'})| .
\]
All coefficients are rescaled as \(b_{g,r}\leftarrow b_{g,r}/S_{\text{coef}}\) and \(Q_{(g,r),(g',r')}\leftarrow Q_{(g,r),(g',r')}/S_{\text{coef}}\) (with \(S_{\text{coef}}{=}1\) if the sum is zero), and then the penalty is fixed to \(\lambda = 1+\varepsilon\) (by default, \(\varepsilon=10^{-3}\)).
Because the off-diagonal sum of \(|Q|\) is strictly smaller than the total sum used in \(S_{\text{coef}}\), the inequality \eqref{eqn:lambda} is automatically satisfied for any \(\varepsilon>0\).

Up to an additive constant, the Hamiltonian admits the matrix form
\begin{equation}\label{eqn:matrix_form}
H_{\mathrm{CKM}}(y) \;=\; y^\top Q\, y \;+\; c^\top y \;+\; \lambda \sum_{g}\Bigl(1-\sum_{r} y_{g,r}\Bigr)^2 ,
\end{equation}
where \(y\in\{0,1\}^{n}\) stacks all \(y_{g,r}\) with \(n=\sum_{g=1}^k D_g\), \(Q\) is the real symmetric matrix with entries \(Q_{(g,r),(g',r')}\), and \(c\) stacks \(\{b_{g,r}\}\) in the same order. For quantum optimization we deploy a QAOA ansatz with a one-hot mixer that is invariant on the subspace defined by \(\sum_{r=1}^{D_g} y_{g,r}=1\) for each \(g\in\mathcal{K}\), so that the penalty Hamiltonian remains inactive throughout the variational evolution.

\begin{figure*}[t]
\centering

\begin{minipage}[t]{0.485\textwidth}
\centering
\begin{quantikz}[row sep=0.3cm, column sep=0.7cm]
  \lstick{$\ket{0}$} & \gate{H} & \gate[wires=3]{\Ulabel} & \qw & \qw \\
  \lstick{$\ket{0}$} & \gate{H} & \ghost{\Ulabel}        & \qw & \qw \\
  \lstick{$\ket{0}$} & \gate{H} & \ghost{\Ulabel}        & \qw & \qw \\
  \lstick{$\ket{0}$} & \gate{H} & \ctrl{-3}              & \gate{H} & \meter{} \\
\end{quantikz}

\vspace{0.35em}
\footnotesize\emph{(a) Real part.}
\end{minipage}
\hfill
\begin{minipage}[t]{0.485\textwidth}
\centering
\begin{quantikz}[row sep=0.3cm, column sep=0.7cm]
  \lstick{$\ket{0}$} & \gate{H} & \gate[wires=3]{\Ulabel} & \qw & \qw & \qw \\
  \lstick{$\ket{0}$} & \gate{H} & \ghost{\Ulabel}        & \qw & \qw & \qw \\
  \lstick{$\ket{0}$} & \gate{H} & \ghost{\Ulabel}        & \qw & \qw & \qw \\
  \lstick{$\ket{0}$} & \gate{H} & \ctrl{-3}              & \gate{S^\dagger} & \gate{H} & \meter{} \\
\end{quantikz}

\vspace{0.35em}
\footnotesize\emph{(b) Imaginary part.}
\end{minipage}

\caption{Hadamard-test circuits for estimating the complex quantity 
$\mu^{(j)}_M=\langle \psi \mid U \mid \psi \rangle$. 
(a) \emph{Real part:} the $n_i=3$ index qubits are prepared in a uniform superposition and control the diagonal oracle $U=\mathrm{diag}(e^{i\theta_1},\dots,e^{i\theta_B},1,\dots,1)$. 
Measuring the ancilla in the $X$ basis yields $\Re\,\mu^{(j)}_M=\mathbb{E}[(-1)^Z]$. 
(b) \emph{Imaginary part:} inserting $S^\dagger$ on the ancilla before the final $H$ implements a $Y$-basis measurement, giving $\Im\,\mu^{(j)}_M=\mathbb{E}[(-1)^Z]$ under our sign convention.}
\label{fig:qff-both}
\end{figure*}
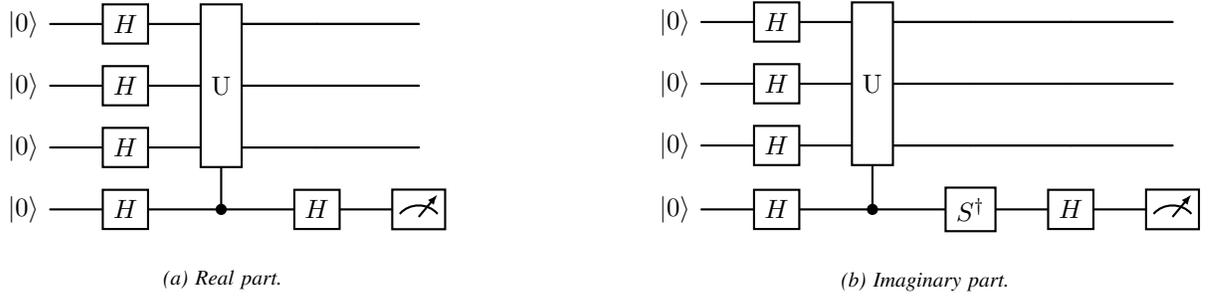

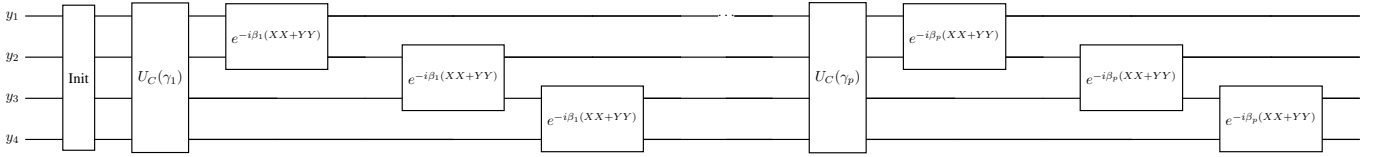
\begin{figure*}[t]
\centering
\scalebox{0.55}{
\begin{quantikz}[row sep=0.35cm, column sep=0.9cm]
  \lstick{$y_1$} & \gate[wires=4]{\text{Init}} & \gate[wires=4]{U_C(\gamma_1)}
    & \gate[2]{e^{-i\beta_1(XX{+}YY)}} & \qw
    & \qw & \qw & \qw
    & \push{\cdots} & \qw
    & \gate[wires=4]{U_C(\gamma_p)}
    & \gate[2]{e^{-i\beta_p(XX{+}YY)}} & \qw
    & \qw & \qw & \qw \\
  \lstick{$y_2$} & \ghost{\text{Init}} & \ghost{U_C(\gamma_1)}
    & \ghost{e^{-i\beta_1(XX{+}YY)}} & \qw
    & \gate[2]{e^{-i\beta_1(XX{+}YY)}} & \qw & \qw
    & \push{} & \qw
    & \ghost{U_C(\gamma_p)}
    & \ghost{e^{-i\beta_p(XX{+}YY)}} & \qw
    & \gate[2]{e^{-i\beta_p(XX{+}YY)}} & \qw & \qw \\
  \lstick{$y_3$} & \ghost{\text{Init}} & \ghost{U_C(\gamma_1)}
    & \qw & \ghost{e^{-i\beta_1(XX{+}YY)}}
    & \ghost{e^{-i\beta_1(XX{+}YY)}} & \gate[2]{e^{-i\beta_1(XX{+}YY)}} & \qw
    & \push{} & \qw
    & \ghost{U_C(\gamma_p)} & \qw
    & \ghost{e^{-i\beta_p(XX{+}YY)}} & \ghost{e^{-i\beta_p(XX{+}YY)}}
    & \gate[2]{e^{-i\beta_p(XX{+}YY)}} & \qw \\
  \lstick{$y_4$} & \ghost{\text{Init}} & \ghost{U_C(\gamma_1)}
    & \qw & \qw & \qw
    & \ghost{e^{-i\beta_1(XX{+}YY)}} & \qw
    & \push{} & \qw
    & \ghost{U_C(\gamma_p)} & \qw & \qw & \qw
    & \ghost{e^{-i\beta_p(XX{+}YY)}} & \qw \\
\end{quantikz}
}
\caption{Depth-$p$ per-group QAOA for $D_g=4$. Each layer $\ell$ applies $U_C(\gamma_\ell)$
followed by the XY mixer in two sublayers of adjacent pairs: $(1,2)$ and $(3,4)$, then $(2,3)$.
For a complete ring, the pair $(4,1)$ can be implemented with a SWAP network or an additional column.}
\label{fig:qaoa-group-p}
\end{figure*}

\section{Quantum Variational Framework} \label{sec:qvf}

We estimate the feature-space centroid $z_X$ introduced in Section~\ref{sec:ckm_problem} using subsampled Quantum Fourier Features (QFF) to obtain unbiased estimates of each component with reduced qubit counts.

\begin{assumption}[Bounded data and uniform subsampling]\label{asmp:bounded}
The dataset $X$ lies within a bounded subset of $\mathbb{R}^d$ and, when estimating $z_X$ via QFF, subsampling is uniform without replacement.
\end{assumption}

\noindent
For the specific case of QFF with complex phases $V_i=e^{\mathrm{i}w^\top x_i}$ one has $\lvert V_i\rvert=1$, so population variances are bounded. Assumption~\ref{asmp:bounded} is not strictly needed for the variance bounds below in this setting, but we keep it to cover more general feature maps.

\subsection{Quantum estimation of $z_X$ via subsampled QFF}\label{subsec:qff}
Each component of $z_X$ in \eqref{eqn:features_data} is the empirical mean of phases $e^{\mathrm{i}w_j^\top x_i}$. 
We estimate these means on quantum hardware with a Hadamard-test circuit using a controlled diagonal oracle 
$\mathrm{diag}\!\big(e^{\mathrm{i}w_j^\top x_1},\ldots,e^{\mathrm{i}w_j^\top x_N}\big)$.
To keep the index register small on NISQ devices, we subsample $B\ll N$ items \emph{per feature component} and pad to a power of two: let
\[
n_i=\max\{1,\lceil\log_2 B\rceil\},\qquad M=2^{n_i}\ (\ge 2).
\]
We then run two Hadamard tests per component (one per measurement basis) and combine them into a single complex quantity
\begin{align}
\mu^{(j)}_M &= \langle \psi|U|\psi\rangle \in \mathbb{C}, \\
\Re\,\mu^{(j)}_M &= \mathbb{E}_{\text{X-basis}}[(-1)^Z], \\
\Im\,\mu^{(j)}_M &= \mathbb{E}_{\text{Y-basis}}[(-1)^Z],
\end{align}
where $Z\in\{0,1\}$ is the ancilla measurement in a Hadamard test with control-unitary $U$ (defined below) and index-register state $\lvert \psi\rangle$. 
Real and imaginary parts are obtained by measuring the ancilla in the $X$- and $Y$-basis, respectively; operationally, we apply an $S^\dagger$ on the ancilla \emph{before} the final Hadamard to measure in the $Y$-basis for the imaginary part, ensuring the sign convention $\langle Y\rangle = \Im\langle\psi|U|\psi\rangle$. The feature-compression register uses $n_i$ index qubits plus one ancilla, so its peak is $n_i+1$ qubits. Statistical error decreases with shots, while subsampling contributes $\mathcal{O}(B^{-1/2})$ Monte Carlo standard error. Throughout, we do not assume QRAM; the diagonal unitary $U$ is compiled from classically computed phases for the chosen subsample.
The unbiased estimator for the $j$-th component is
\begin{equation}\label{eq:qff_estimator}
\hat z_X^{(j)} \;=\; \frac{M}{B}\, \mu^{(j)}_M \;-\; \frac{M-B}{B},
\end{equation}
with $\mu^{(j)}_M$ defined as above from the two runs. 

\begin{proposition}[Unbiased QFF estimator]\label{prop:unbiased}
With padding to $M=2^{\lceil\log_2 B\rceil}$ (and $n_i\ge 1$) and the estimator in \eqref{eq:qff_estimator}, 
$\mathbb{E}\,\hat z_X^{(j)}= z_X^{(j)}$ for both real and imaginary parts.
\end{proposition}

\begin{proof}[Proof of Proposition~\ref{prop:unbiased}]
Fix a feature index $j$ and let $\theta_\ell = w_j^\top x_{i_\ell}$ denote the phases of a uniformly drawn subsample without replacement of size $B$.
We construct $U=\mathrm{diag}(e^{\mathrm{i}\theta_1},\ldots,e^{\mathrm{i}\theta_B}, \underbrace{1,\ldots,1}_{M-B})\in\mathbb{C}^{M\times M}$ and prepare the index register in 
$\lvert \psi\rangle = \frac{1}{\sqrt{M}}\sum_{m=1}^M \lvert m\rangle$.
The Hadamard test with an ancilla yields
\[
\mathbb{E}[(-1)^Z]
=\begin{cases}
\mathrm{Re}\,\langle \psi|U|\psi\rangle, & \text{(no $S$)}\\[2pt]
\mathrm{Im}\,\langle \psi|U|\psi\rangle, & \text{(with $S^\dagger$ before the final $H$).}
\end{cases}
\]
Combining both runs (real and imaginary parts), we obtain
\[
\mu_M \;=\; \langle \psi|U|\psi\rangle
= \frac{1}{M}\!\left(\sum_{\ell=1}^B e^{\mathrm{i}\theta_\ell} + (M-B)\right).
\]
By the definition of the estimator,
\[
\hat z_X^{(j)} \;=\; \frac{M}{B}\,\mu_M \;-\; \frac{M-B}{B}
= \frac{1}{B}\sum_{\ell=1}^B e^{\mathrm{i}\theta_\ell},
\]
which is exactly the sample mean (over the subsample) of the true phases.
Since the subsample is drawn uniformly without replacement, 
\[
\mathbb{E}_{\text{subsamp.}}\!\left[\frac{1}{B}\sum_{\ell=1}^B e^{\mathrm{i}w_j^\top x_{i_\ell}}\right]
= \frac{1}{N}\sum_{i=1}^N e^{\mathrm{i}w_j^\top x_i}
= z_X^{(j)}.
\]
Finally, the shot noise of the Hadamard test is unbiased (the expectation of the estimate of $\mu_M$ equals $\mu_M$), hence 
$\mathbb{E}\,\hat z_X^{(j)}=z_X^{(j)}$ for both real and imaginary parts. 
\end{proof}

\begin{lemma}[Shot-noise variance]\label{lem:shots}
Let $\widehat{\mu}^{(j)}_M$ be estimated from $S$ Hadamard-test shots \emph{per measurement basis} (real or imaginary part). Then
$\mathrm{Var}(\widehat{\mu}^{(j)}_M)\le S^{-1}$ and 
$\mathrm{Var}\!\left(\frac{M}{B}\widehat{\mu}^{(j)}_M\right)\le \frac{M^2}{B^2 S}$ per component.
\end{lemma}

\begin{proof}[Proof of Lemma~\ref{lem:shots}]
In a single shot, $Z\in\{0,1\}$ and $Y:=(-1)^Z\in\{-1,+1\}$ with $\mathbb{E}[Y]=\mu_M$.
From $S$ shots, the estimator $\widehat{\mu}_M=\frac{1}{S}\sum_{s=1}^S Y_s$ satisfies
\[
\mathrm{Var}(\widehat{\mu}_M)=\frac{1}{S}\,\mathrm{Var}(Y)\le \frac{1}{S},
\]
since $\mathrm{Var}(Y)=1-\mu_M^2\le 1$. 
Multiplying by $M/B$ gives 
$\mathrm{Var}\!\left(\frac{M}{B}\widehat{\mu}_M\right)\le \frac{M^2}{B^2 S}$ per real/imag component.
\end{proof}

\begin{lemma}[Subsampling variance]\label{lem:subsample}
Under Assumption~\ref{asmp:bounded}, the subsampling induces 
\[
\mathrm{Var}\!\left(\frac{M}{B}\mu^{(j)}_M - \frac{M-B}{B}\right)\;\le\; \frac{c_j}{B},
\]
with a constant $c_j$ depending on the phase dispersion (finite-population correction applies).
\end{lemma}

\begin{proof}[Proof of Lemma~\ref{lem:subsample}]
The quantity $\frac{M}{B}\mu^{(j)}_M-\frac{M-B}{B}$ coincides exactly with the sample mean without replacement 
\(
\bar{V}_B=\frac{1}{B}\sum_{\ell=1}^B V_{i_\ell}
\)
of the finite population $\{V_i\}_{i=1}^N$ with $V_i=e^{\mathrm{i}w_j^\top x_i}$. For a finite population and uniform sampling without replacement,
\[
\mathrm{Var}(\bar{V}_B)=\frac{1-f}{B}\,S_V^2, 
\qquad f=\frac{B}{N},
\]
where $S_V^2=\frac{1}{N-1}\sum_{i=1}^N \lvert V_i-\bar V\rvert^2$ is the (complex) population variance in squared norm and $\bar V=\frac{1}{N}\sum_{i=1}^N V_i$. Hence,
\[
\mathrm{Var}(\bar{V}_B)\;=\;\frac{1-f}{B}\,S_V^2 \;\le\; \frac{S_V^2}{B}\;=\;\frac{c_j}{B},
\]
with the \emph{constant} $c_j:=S_V^2$ independent of $B$. In particular, since $|V_i|=1$,
\[
S_V^2 \;=\; \frac{N}{N-1}\Big(1-|\bar V|^2\Big)
\;\;\le\;\; \frac{N}{N-1}
\;\;\le\;\; 2,
\]
so $c_j=\frac{N}{N-1}\big(1-|\bar V|^2\big)$ and $\mathrm{Var}(\bar{V}_B)\le \tfrac{1}{B}\cdot \tfrac{N}{N-1}$.
\end{proof}

\begin{corollary}[MSE decomposition and sample complexity]\label{cor:mse}
$\mathrm{MSE}(\hat z_X^{(j)})\le \frac{c_j}{B}+\frac{2M^2}{B^2 S}$, where the factor $2$ accounts for real and imaginary parts measured separately with $S$ shots each. 
Choosing $B=\Theta(\varepsilon^{-2})$ and $S=\Theta(\varepsilon^{-2})$ yields 
$\mathrm{MSE}(\hat z_X^{(j)})=\mathcal{O}(\varepsilon^2)$.
\end{corollary}

\begin{proof}[Proof of Corollary~\ref{cor:mse}]
By the law of total variance, conditioning on the subsample,
\begin{align}
\mathrm{MSE}(\hat z_X^{(j)}) &= \mathrm{Var}(\hat z_X^{(j)}) \nonumber \\
&= \mathbb{E}\!\left[\mathrm{Var}\!\left(\hat z_X^{(j)} \mid \text{subsamp.}\right)\right] \nonumber \\
&\quad + \mathrm{Var}\!\left(\mathbb{E}\!\left[\hat z_X^{(j)} \mid \text{subsamp.}\right]\right).
\end{align}
The first term is the shot-noise variance of the Hadamard test, bounded by Lemma~\ref{lem:shots} as $\frac{2M^2}{B^2 S}$ when combining real and imaginary parts; the second is the variance due to subsampling, bounded by Lemma~\ref{lem:subsample} as $\frac{c_j}{B}$. 
Adding both,
\[
\mathrm{MSE}(\hat z_X^{(j)}) \;\le\; \frac{c_j}{B} \;+\; \frac{2M^2}{B^2 S}.
\]
Since we pad $B$ to the next power of two, we have $B\le M<2B$, hence $\frac{M}{B}\le 2$ and the shot-noise contribution is uniformly bounded as $\frac{2M^2}{B^2 S}\le \frac{8}{S}$. 
Choosing $B=\Theta(\varepsilon^{-2})$ and $S=\Theta(\varepsilon^{-2})$ yields 
$\mathrm{MSE}(\hat z_X^{(j)})=\mathcal{O}(\varepsilon^2)$.
\end{proof}

\subsection{Quantum solver with one-hot preserving mixer}
We solve \eqref{eqn:qubo_ckm} using QAOA with a mixer that preserves the one-hot subspace per cluster.
For a cluster with $D$ candidates we use the XY ring mixer
\begin{equation}\label{eq:xy_ring}
\begin{aligned}
H_{\mathrm{mix}}^{(D)} &= \sum_{t=1}^{D} \bigl( X_t X_{t+1} + Y_t Y_{t+1} \bigr), \\
&\quad X_{D+1}\equiv X_1,\quad Y_{D+1}\equiv Y_1.
\end{aligned}
\end{equation}
which conserves Hamming weight and thus keeps exactly one excitation. 
Initial states are either the $D$-qubit $W$-state or the computational state $\lvert 100\ldots 0\rangle$, depending on $D$. 
We adopt a \emph{per-group} approach: each cluster defines its own $D_g$-variable QUBO solved independently, capping the decision-register size to $\max_g D_g$ instead of $\sum_g D_g$.

\begin{lemma}[One-hot invariance]\label{lem:onehot}
For $D\ge 2$, the XY ring mixer in \eqref{eq:xy_ring} preserves the Hamming weight-one subspace on $D$ qubits. (For $D=1$ the claim is trivial.)
\end{lemma}

\begin{proof}[Proof of Lemma~\ref{lem:onehot}]
Let $N:=\sum_{t=1}^D \frac{1-Z_t}{2}$ be the excitation-number operator (Hamming weight). 
It suffices to show $[H_{\mathrm{mix}}^{(D)},N]=0$. Since $N=\frac{D}{2}-\frac12\sum_t Z_t$, this is equivalent to showing $[H_{\mathrm{mix}}^{(D)},\sum_t Z_t]=0$.
Consider a single edge term $h_{t,t+1}:=X_tX_{t+1}+Y_tY_{t+1}$. Using the Pauli commutation rules one checks
$[h_{t,t+1},Z_s]=0$ for $s\notin\{t,t+1\}$, and
$[h_{t,t+1},Z_t+Z_{t+1}]=0$.
Therefore $[h_{t,t+1},\sum_s Z_s]=0$ for each $t$, and by linearity
$[H_{\mathrm{mix}}^{(D)},\sum_s Z_s]=0$, hence $[H_{\mathrm{mix}}^{(D)},N]=0$.
Thus $H_{\mathrm{mix}}^{(D)}$ preserves every fixed-Hamming-weight subspace, in particular the weight-one subspace.
\end{proof}

\begin{lemma}[Connectivity]\label{lem:connect}
Restricted to the one-hot basis $\{\lvert 10\ldots0\rangle,\ldots,\lvert 0\ldots01\rangle\}$, 
the mixer graph of $H_{\mathrm{mix}}^{(D)}$ is a single $D$-cycle. Hence, within the one-hot subspace, QAOA can explore all one-hot states.
\end{lemma}

\begin{proof}[Proof of Lemma~\ref{lem:connect}]
Let $|e_r\rangle$ denote the one-hot basis state with the unique excitation at position $r\in\{1,\dots,D\}$.
For a single edge $(t,t+1)$ one has
\begin{align}
(X_tX_{t+1}+Y_tY_{t+1})\,|e_t\rangle &= 2\,|e_{t+1}\rangle, \\
(X_tX_{t+1}+Y_tY_{t+1})\,|e_{t+1}\rangle &= 2\,|e_{t}\rangle
\end{align}

and the action is zero on $|e_r\rangle$ for $r\notin\{t,t+1\}$.
Summing over all edges on the ring, $H_{\mathrm{mix}}^{(D)}$ couples $|e_r\rangle$ only to its two neighbors $|e_{r-1}\rangle$ and $|e_{r+1}\rangle$ (indices mod $D$), with nonzero matrix elements. The induced interaction graph is thus the cycle $C_D$, which is connected.
\end{proof}

\begin{proposition}[Per-group decomposition]\label{prop:pergroup}
If inter-group couplings are dropped (as in our implementation), the objective decouples into
$k$ independent $D_g$-variable QUBOs—one per cluster—and the joint minimizer is obtained by
concatenating the per-group solutions.
\end{proposition}

\begin{proof}[Proof of Proposition~\ref{prop:pergroup}]
Let the full (penalized) objective be $F(y^{(1)},\dots,y^{(k)})=\sum_{g=1}^k f_g\!\left(y^{(g)}\right)$, 
with one-hot constraints applied independently to each block $y^{(g)}\in\{0,1\}^{D_g}$. 
Since the feasible set factors as a Cartesian product and $F$ is additively separable,
\[
\arg\min_{(y^{(1)},\dots,y^{(k)})} \sum_{g=1}^k f_g\!\left(y^{(g)}\right)
= \prod_{g=1}^k \arg\min_{y^{(g)}} f_g\!\left(y^{(g)}\right).
\]
Thus any concatenation of per-group minimizers is a global minimizer, and conversely any global minimizer concatenates per-group minimizers.
\end{proof}

\noindent\emph{Notes.} 
(i) The one-hot mixer preserves the feasible subspace; we keep a scaled one-hot penalty in the QUBO to stabilize optimization under noise and mild model mismatch.
(ii) We apply a simple coefficient normalization in practice to improve numerical conditioning; this scaling does not change the minimizer but affects the effective optimization step sizes.

\subsection{q-Lloyd refinement}
After the initial per-group selection, we perform a few outer iterations to refine centroids:
(i) reassign points to the current centroids (nearest in the \emph{standardized} input space),
(ii) recompute per-cluster feature means $z_{X,g}$ via subsampled QFF restricted to the cluster,
(iii) rebuild each group QUBO with candidates jittered around the current centroid, and
(iv) re-solve with QAOA.
\paragraph{Elitist retention.}
At each outer iteration we include, for every group $g$, the centroid selected in the previous iteration among the new $D_g$ candidates. 
This guarantees that the previous per-group cost value is attainable in the new subproblem, enabling a direct comparison.
This q-Lloyd loop typically converges in a handful of steps under a tolerance on centroid movement.

\begin{proposition}[Descent of the conditional step under elitist retention]\label{prop:descent}
Assume \emph{elitist retention} and fix the assignments. 
For each group $g$, let $\mathcal{C}_g=\{\varphi(c_{g,r})\}_{r=1}^{D_g}$ be the candidate feature vectors and $z_{X,g}$ the current cluster feature mean.
If each per-group QUBO is solved to $\delta$-optimality (i.e., its returned value is within $\delta\ge 0$ of the group minimum), then the global surrogate
\(
\|z_X - z_\mu\|^2 = \sum_{g=1}^k \min_{r}\|\varphi(c_{g,r})-z_{X,g}\|^2
\)
weakly decreases up to an $\mathcal{O}(\delta)$ term when replacing the previous candidates by the new per-group selections.
\end{proposition}

\begin{proof}[Proof of Proposition~\ref{prop:descent}]
With assignments fixed, the CKM surrogate decomposes as
\[
\|z_X - z_\mu\|^2 \;=\; \sum_{g=1}^k \underbrace{\min_{r\in[D_g]}\|\varphi(c_{g,r})-z_{X,g}\|^2}_{=:F_g^\star}.
\]
Let $\tilde r_g$ be the index returned by the (approximate) QAOA for group $g$, and $F_g(\tilde r_g)$ its attained value. By $\delta$-optimality, $F_g(\tilde r_g)\le F_g^\star+\delta$ for all $g$.
Under elitist retention, the previous selection $r_g^{\mathrm{old}}$ is included among the current candidates, so $F_g^\star\le F_g(r_g^{\mathrm{old}})$. 
Hence
\(
F_g(\tilde r_g) \;\le\; F_g^\star+\delta \;\le\; F_g(r_g^{\mathrm{old}})+\delta.
\)
Summing over $g$ gives the claimed global weak decrease up to $k\delta=\mathcal{O}(\delta)$.
\end{proof}

\begin{remark}[Reassignment step]
The nearest-centroid reassignment is performed in standardized input space and may not strictly 
decrease the feature-space objective; empirically, the outer loop converges in few iterations.
\end{remark}

\begin{corollary}[Termination under tolerance]\label{cor:stop}
With fixed QAOA depth and shot budget across iterations, the q-Lloyd loop stops in finitely many iterations once $\|C^{(t)}-C^{(t-1)}\|\le\tau$. 
A stronger (sufficient) condition for monotone improvement is to increase the subproblem accuracy over iterations (e.g., larger shot budgets or depth), but this is not required by our implementation.
\end{corollary}

\begin{proof}[Proof of Corollary~\ref{cor:stop}]
The algorithm explicitly tests the stopping condition $\|C^{(t)}-C^{(t-1)}\|\le\tau$ at each iteration $t$. Therefore it halts at the first $t$ for which the condition holds. No monotonicity of the objective is required for this termination guarantee.
\end{proof}

\subsection{Hardware scaling and error analysis} \label{sec:scaling-error}
The grouped QUBO formulation in qc-kmeans reduces the peak qubit requirement from 
approximately $kD$ in the joint CKM approach to
\begin{subequations}\label{eq:qubit_peak}
\begin{align}
q_{\mathrm{peak}} &= \max\{D,\;\lceil\log_2 B\rceil_+ + 1\}, \label{eq:qubit_peak_main}\\
\lceil\log_2 B\rceil_+ &:= \max\{1,\lceil\log_2 B\rceil\}. \label{eq:qubit_peak_def}
\end{align}
\end{subequations}
where $D = \max_g D_g$ is the largest number of candidate centroids considered in any cluster 
and $B$ is the QFF subsample size.  
This per-group decomposition enables shallow-depth QAOA circuits (1–2 layers) compatible with NISQ hardware.

\begin{proposition}[Peak-qubit bound]\label{prop:qpeak}
Let $D = \max_g D_g$ and $n_i := \max\{1,\lceil\log_2 B\rceil\}$.  
Then $q_{\mathrm{peak}} = \max\{D, n_i+1\}$ for the per-group scheme,  
while a joint CKM QUBO requires at least $kD$ decision qubits.
\end{proposition}

\begin{proof}[Proof of Proposition~\ref{prop:qpeak}]
In the per-group scheme, QAOA over a single group requires exactly $D_g$ decision qubits; across all groups the largest such decision register is $D=\max_g D_g$. Separately, QFF estimation with subsample size $B$ uses an index register of $n_i=\max\{1,\lceil\log_2 B\rceil\}$ qubits plus one ancilla (for the Hadamard test), amounting to $n_i+1$ qubits. Since these two routines are executed in separate circuits (never concurrently), the peak number of qubits is the maximum of the two counts, i.e., $\max\{D,n_i+1\}$. 
By contrast, the joint CKM QUBO packs all groups into a single decision register with $\sum_g D_g\ge kD$ qubits, hence it requires at least $kD$ decision qubits.
\end{proof}

\begin{lemma}[Mixer depth and gate count]\label{lem:depth}
With $p$ QAOA layers and XY ring connectivity, the two-qubit gate count \emph{of the mixer} per group 
scales as $\mathcal{O}(pD)$; the hidden constant depends on the chosen decomposition of the XY mixer into native gates.
Separately, the \emph{cost} layer for a $D_g$-variable QUBO contributes $\mathcal{O}\!\big(p\,\mathrm{nnz}(Q_g)\big)$
two-qubit gates, which is $\mathcal{O}(pD_g^2)$ in the worst case and $\mathcal{O}(pD_g)$ if the quadratic form is ring-sparse.
\end{lemma}

\begin{proof}[Proof of Lemma~\ref{lem:depth}]
A single XY ring layer comprises $D$ two-qubit interactions (one per edge of the cycle). Each such interaction compiles
into a constant number of native two-qubit gates, hence a mixer layer is $\mathcal{O}(D)$ and $p$ layers give $\mathcal{O}(pD)$.
For the cost unitary $\exp(-i\gamma H_{\mathrm{cost}})$ with quadratic Ising form
$H_{\mathrm{cost}}=\sum_i a_i Z_i + \sum_{i<j} b_{ij} Z_i Z_j$, the number of two-qubit entangling gates scales with the number
of nonzero couplings $\mathrm{nnz}(Q_g)=|\{(i,j): b_{ij}\neq 0\}|$. Therefore each QAOA layer contributes
$\mathcal{O}(\mathrm{nnz}(Q_g))$ two-qubit gates for the cost, i.e., $\mathcal{O}\!\big(p\,\mathrm{nnz}(Q_g)\big)$ in $p$ layers.
In the dense worst case $\mathrm{nnz}(Q_g)=\Theta(D_g^2)$; for ring-sparse costs, $\mathrm{nnz}(Q_g)=\Theta(D_g)$.
\end{proof}

\noindent
\textbf{Relaxation error of the per-group QUBO:}
Let the full CKM QUBO be $F(x)=x^\top Q x + c^\top x$ over the feasible set $\Omega$ of concatenated one-hot vectors $x=[x^{(1)};\ldots;x^{(k)}]$, where $x^{(g)}\in\{0,1\}^{D_g}$ and $\mathbf{1}^\top x^{(g)}=1$.
Partition $Q$ into diagonal blocks $Q_g$ (intra-group terms) and off-diagonal blocks $R_{gh}$ ($g\neq h$), so that
\[
Q=\mathrm{blkdiag}(Q_1,\ldots,Q_k)+\sum_{g<h}\big(E_{gh}(R_{gh})+E_{hg}(R_{gh}^\top)\big),
\]
with $E_{gh}(\cdot)$ embedding a block at $(g,h)$.
The grouped surrogate drops the inter-group couplings:
\[
\tilde F(x)=\sum_{g=1}^k (x^{(g)})^\top Q_g\, x^{(g)} + \sum_{g=1}^k c_g^\top x^{(g)}.
\]
Define the relaxation term $E(x):=F(x)-\tilde F(x)=\sum_{g<h}\!\big((x^{(g)})^\top R_{gh} x^{(h)}+(x^{(h)})^\top R_{hg} x^{(g)}\big)$.

\begin{lemma}[Pointwise bound]\label{lem:pointwise_relax}
For any feasible $x\in\Omega$,
\[
|E(x)| \;\le\; 2\sum_{g<h}\|R_{gh}\|_{\infty},
\]
where $\|\cdot\|_{\infty}$ is the entrywise max-norm.
\end{lemma}

\begin{proof}
Each $x^{(g)}$ is one-hot, so $(x^{(g)})^\top R_{gh} x^{(h)}$ selects a single entry of $R_{gh}$; its magnitude is at most $\|R_{gh}\|_{\infty}$. The term with $R_{hg}=R_{gh}^\top$ contributes the same bound. Summing over $g<h$ gives the claim.
\end{proof}

\begin{theorem}[Suboptimality of the grouped relaxation]\label{thm:gap_relax}
Let $x^\star\in\arg\min_{x\in\Omega}F(x)$ and $\tilde x\in\arg\min_{x\in\Omega}\tilde F(x)$. Then
\(
 0 \;\le\; F(\tilde x)-F(x^\star) \;\le\; 2\max_{x\in\Omega}|E(x)| \;\le\; 4\sum_{g<h}\|R_{gh}\|_{\infty}.
\)
\end{theorem}

\begin{proof}
Let $F(\tilde x)=\tilde F(\tilde x)+E(\tilde x)\le \tilde F(x^\star)+|E(\tilde x)|$ by optimality of $\tilde x$ for $\tilde F$. Also $F(x^\star)=\tilde F(x^\star)+E(x^\star)$. Subtracting yields $F(\tilde x)-F(x^\star)\le |E(\tilde x)|-|E(x^\star)|\le |E(\tilde x)|+|E(x^\star)|\le 2\max_{x\in\Omega}|E(x)|$. Apply Lemma~\ref{lem:pointwise_relax}.
\end{proof}

\noindent
\textit{Remark.} If $R_{gh}=0$ for all $g\neq h$ (block-diagonal $Q$), then $E(x)\equiv 0$ and the grouped and joint formulations coincide.
More generally, if for each $g$ one has $\sum_{h\neq g}\|R_{gh}\|_{\infty}\le \varepsilon$, then
\(
\sum_{g<h}\|R_{gh}\|_{\infty}\;\le\;\tfrac12\sum_{g\neq h}\|R_{gh}\|_{\infty}\;\le\;\tfrac12\sum_g \varepsilon \;=\; \tfrac{k}{2}\,\varepsilon,
\)
and Theorem~\ref{thm:gap_relax} gives
\(
F(\tilde x)-F(x^\star)\;\le\;4\sum_{g<h}\|R_{gh}\|_{\infty}\;\le\;2k\,\varepsilon,
\)
hence the relaxation gap is $\mathcal{O}(k\,\varepsilon)$.

\medskip

These hardware savings come with approximation effects. Random Fourier features approximate shift-invariant kernels with error $\mathcal{O}(m^{-1/2})$, where $m$ is the number of features. The CKM surrogate $\|z_X - z_\mu\|^2$ measures feature-space alignment; stochastic error from Quantum Fourier Feature estimation with subsample size $B$ contributes $\mathcal{O}(B^{-1/2})$ standard error plus quantum shot noise. \emph{In addition}, the per-group relaxation introduces a \emph{deterministic} approximation term $E(x)$ bounded by Theorem~\ref{thm:gap_relax}; this term vanishes when inter-group couplings are negligible or block-diagonal.

\subsection{Quantum Compressive  $k$-Means Clustering Scheme} \label{sec:method:overview}
We adopt the Quantum Compressive  $k$-Means (qc-kmeans) approach, which combines compressed feature representation with quantum optimization for centroid selection and refinement. The method applies to unsupervised clustering under limited qubit resources, using sub-sampled quantum feature estimation and independent per-cluster QUBO solves. Algorithm~\ref{alg:bb_sche} summarizes the procedure.

\begin{algorithm}[tbh]
\caption{Quantum Compressive $k$-Means (QFF + grouped QAOA)}
\label{alg:bb_sche}
\begin{algorithmic}[1]
  \Require Data $X\in\mathbb{R}^{N\times d}$, clusters $k$, \#freq. $m$, cand./cluster $D$, QAOA depth $p$, penalty $\lambda$, subsample $B$
  \Ensure Centroids $\hat{C}$, assignments $\hat{y}$
  \State $X_s \gets \mathrm{StandardScaler}(X)$
  \State Draw $W\in\mathbb{R}^{m\times d}$ with rows $\sim \mathcal{N}(0,I_d)$
  \State \textbf{Estimate $z_X$ (QFF only):} $z_X \gets \mathrm{Lazy\text{-}QFF\_Estimate}(X_s,W,B)$
  \State Seeds $S \gets k$-means$(X_s)$
  \State Generate $D$ jittered candidates $C_g[r]$ per cluster; features $V_g[r] \gets e^{iW C_g[r]}$
  \For{$g=1,\dots,k$}
    \State Solve grouped QUBO $\min_{y_g} \| \sum_r y_{g,r} V_g[r] - z_{X,g}\|^2 + \lambda\,\text{one-hot}$ via QAOA
    \State $\hat{C}_g \gets C_g[\arg\max_r y_{g,r}]$
  \EndFor
  \If{refinement $>0$} \Comment{grouped q-Lloyd}
    \Repeat
      \State Reassign points; resample local candidates around $\hat{C}_g$; recompute $V_g[r]$
      \State Re-solve grouped QUBOs via QAOA; update $\hat{C}_g$
    \Until{convergence or max steps}
  \EndIf
  \State \Return $\hat{C},\hat{y}$
\end{algorithmic}
\end{algorithm}


\section{Computational Experiments} \label{sec:experiments}

We implemented two variants of qc-kmeans in Qiskit~1.3.2 using Python~3.10.12. Each quantum circuit uses a single QAOA layer ($p=1$) and a total of 10{,}000 shots for observable estimation, following the hyperparameter guidelines in~\cite{hao2024end}. 
Circuits were simulated on Qiskit’s \texttt{AerSimulator} with IBM-Q noise emulation. All simulations were executed on an Ubuntu Linux system (kernel~6.8.0-51-generic) with an Intel\textregistered{} Xeon\textregistered{} Gold~6230R CPU @~2.10\,GHz (104 logical cores) and 187\,GiB of RAM.
Our algorithm targets NISQ devices through three design choices: (i) per-group QUBOs to limit the decision register to $D$ qubits, (ii) an XY mixer that preserves the one-hot subspace at low depth, and (iii) sub-sampled QFF to bound the feature-compression register to $n_i+1$ qubits. These constraints define the peak-qubit budget in~\eqref{eq:qubit_peak} and enable an end-to-end quantum pipeline using subsampled QFF($B$) and quantum per-group assignment via QAOA.
\textcolor{black}{
All experiments with \texttt{qc-kmeans} use $m=4kd$, a per–cluster candidate set size $D=4$, QAOA \texttt{reps}$=1$, a total shot budget of $10^4$, and a QFF subsample size $B=2^8=256$.}
\color{black}
We evaluate qc-kmeans and its noise-simulated variant on 15 datasets. Six of these are widely used synthetic benchmarks \texttt{circles}, \texttt{moons}, \texttt{spiral}, \texttt{moons\_2}, \texttt{an blobs}, and \texttt{vd blobs} with sample sizes ranging from \(n\) = 300 to 2{,}100. In addition to \texttt{Hemicellulose} \cite{wang_predicting_2022} and \texttt{PR2392} \cite{padberg_branch-and-cut_1991}, we consider seven real-world datasets from the UCI Machine Learning Repository \cite{Dua:2019}, with \(n\) ranging from nearly 2{,}000 samples up to 434{,}876 samples. The complete code can be found at \url{https://anonymous.4open.science/r/IEEE_qckmeans-5B8C/}.

\subsection{Benchmark against Quantum Baselines} \label{ssec:baselines} 
\begin{table*}[htbp]
\caption{Performance of qc-kmeans on low\text{-}dimensional synthetic datasets ($d=2$) with fixed\text{-}depth QAOA ($p=1$). The sketch $z_X$ is estimated via Lazy\text{-}QFF ($B=256$) with $D=6$ candidates per cluster. The peak\text{-}qubit bound is $q_{\text{peak}}=\max\{D,\lceil\log_2 B\rceil+1\}$. We allocate $10{,}000$ shots to QAOA; QFF uses $1{,}024$ shots per frequency.}
\label{tab:baselines}
\centering
\begin{sc}
\resizebox{\textwidth}{!}{
\begin{tabular}{ll|ccc|ccc|ccc}
\toprule
 \multirow{2}{*}{Dataset}& \multirow{2}{*}{Method}&\multicolumn{3}{c|}{$k$=3} & \multicolumn{3}{c|}{$k$=5} & \multicolumn{3}{c}{$k$=10} \\
 & &SSE & Qubits & Time (s) & SSE & Qubits & Time (s) & SSE & Qubits & Time (s) \\
\midrule
 Circles& q-means& 75.33 &   11&    -&     53.33 &   11&    -&     30.81 &   11&    -\\
 $n=300$
& q-k-Means&71.57&   5&   764.25
&    41.517
&   5&   764.25
&    33.81&   5&   2{,}539.45\\
 $d=2$& QUBO&136.33&   3&   5{,}004.06&    69.63
&   5&   12{,}973.30&   \multicolumn{3}{c}{No solution found}\\
 & qc-kmeans& \textbf{71.47}&   9&    \textbf{27.74}&     \textbf{39.93}&   9&    \textbf{52.43}&     \textbf{20.06}&   9&    \textbf{233.69}\\
\midrule
 Spiral& q-means& 2{,}104.90&   11&    -&     1{,}292.37&   11&    -&     647.13 &   11&    -\\
 $n=300$
& q-k-Means&2{,}338.25&   5&   733.60&    1{,}360.19&   5&   1{,}202.75&    891.49&   5&   2{,}408.5\\
 $d=2$& QUBO&38{,}014.65&   3&   4{,}898.70&    5{,}785.60&   5&   9{,}466.68&    \multicolumn{3}{c}{No solution found}\\
 & qc-kmeans& \textbf{1{,}880.57}&   9&    \textbf{44.31}&     \textbf{1{,}059.09}&   9&    \textbf{38.35}&     \textbf{448.38}&   9&    \textbf{189.21}\\
\midrule
 Moons\_2& q-means& 250.80 &   11&    -&     219.43 &   11&    -&     62.57 &   11&    -\\
 $n=400$
& q-k-Means&256.28&   5&   1{,}011.65&    138.19
&   5&   1{,}646.9&    51.30
&   5&   3{,}480.55\\
 $d=2$& QUBO&1{,}206.42&   3&   7{,}323.78&    \multicolumn{3}{c|}{No solution found} &    \multicolumn{3}{c}{No solution found}\\
 & qc-kmeans& \textbf{247.15}&   9&    \textbf{61.05}&     \textbf{121.48}&   9&    \textbf{38.52}&     \textbf{42.92}&   9&    \textbf{154.63}\\
\midrule
 Blobs-600& q-means& 368.80 &   12&    -&     309.35 &   12&    -&     206.73 &   12&    -\\
 $n=600$& q-k-Means&365.58
&   5&   1{,}458.3&    228.17
&   5&   2{,}454.75&    145.92
&   5&   4{,}899.65\\
 $d=2$& QUBO&\multicolumn{3}{c|}{No solution found}&   \multicolumn{3}{c|}{No solution found}&    \multicolumn{3}{c}{No solution found}\\
 & qc-kmeans& \textbf{356.07}&   9&    \textbf{51.60}&     \textbf{199.54}&   9&    \textbf{85.59}&     \textbf{120.00}&   9&    \textbf{82.67}\\ 
\midrule
 Blobs-1200& q-means&3.03E+04
&   13&   -&    2.91E+04&   13&   -&    1.43E+04&   13&   -\\
 $n=1{,}200$& q-k-Means&1.57E+05
&   5&   3{,}055.35&    1.57E+05
&   5&   5{,}050.05&    1.53E+05
&   5&   10{,}003.6\\
 $d=2$& QUBO& \multicolumn{3}{c|}{No solution found}&    \multicolumn{3}{c|}{No solution found} &    \multicolumn{3}{c}{No solution found}\\
 & qc-kmeans& \textbf{1.75E+04}&   9&     \textbf{39.44}&     \textbf{1.29E+04}&   9&     \textbf{101.35}&     \textbf{7.83E+03}&   9&     \textbf{320.99}\\ 
\midrule
 Blobs-2100& q-means&5.22E+04&   14&   -&    3.52E+04&   14&   -&    3.90E+04
&   14&   -\\
 $n=2{,}100$& q-k-Means&2.62E+05
&   5&   5{,}062.50&    2.62E+05
&   5&   8{,}394.05&    \multicolumn{3}{c}{No solution found} \\
 $d=2$& QUBO&  \multicolumn{3}{c|}{No solution found}&    \multicolumn{3}{c|}{No solution found} &    \multicolumn{3}{c}{No solution found}\\
 & qc-kmeans& \textbf{2.09E+04}&   9&     \textbf{92.01}&     \textbf{1.58E+04}&   9&     \textbf{82.32}&     \textbf{9.01E+03}&   9&     \textbf{264.77}\\ 
\bottomrule
\end{tabular}
}
\end{sc}
\end{table*}

\begin{figure*}[htbp]
    \centering
    \includegraphics[width=\linewidth]{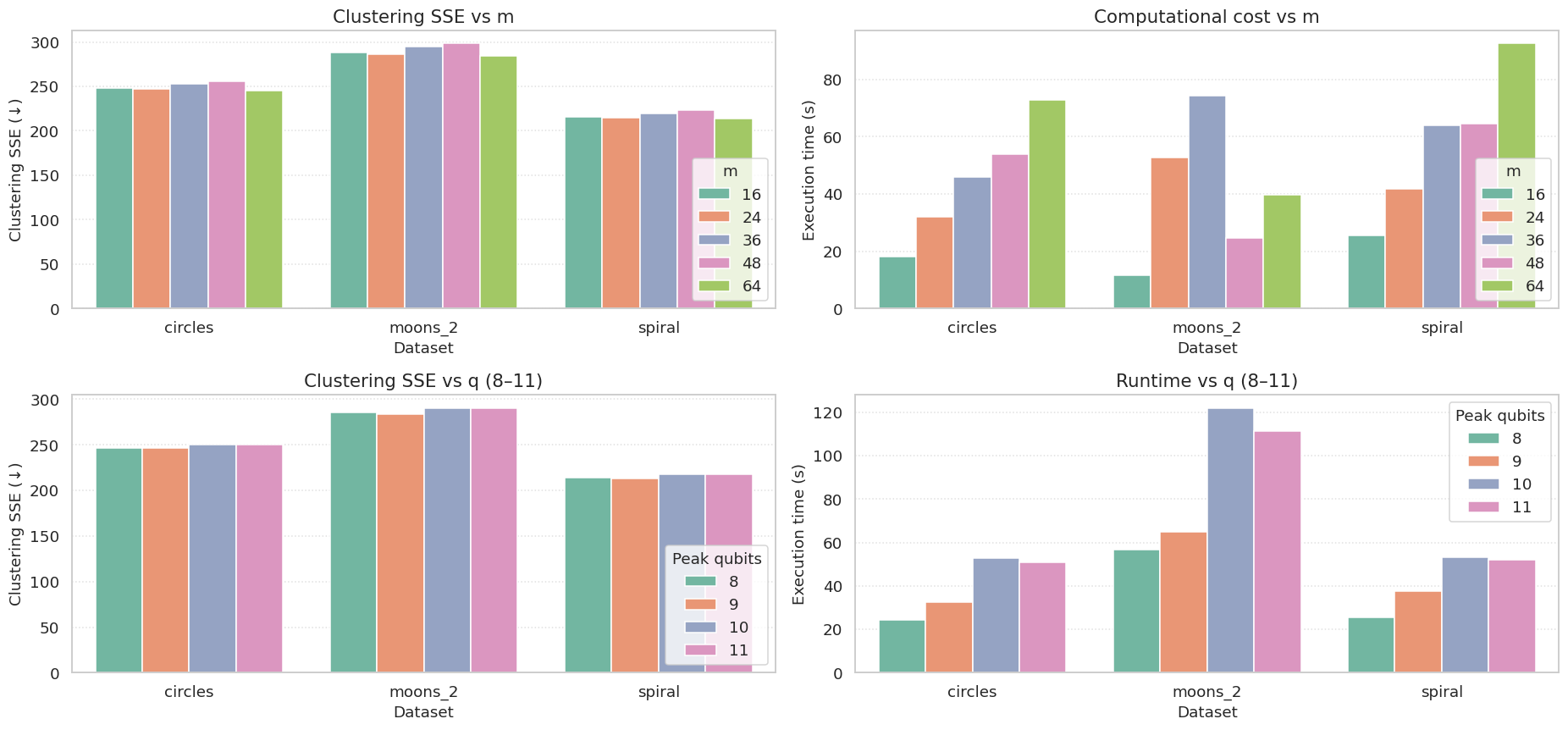}
    \caption{Comparison of clustering SSE and computational cost across datasets for two parameterizations: \textbf{(a)} clustering SSE vs.\ frequency parameter $m$, \textbf{(b)} execution time vs.\ $m$, \textbf{(c)} clustering SSE vs.\ peak number of qubits $q_{\mathrm{peak}}$, \textbf{(d)} execution time vs.\ $q_{\mathrm{peak}}$. For subfigures (c) and (d), only $q_{\mathrm{peak}} \in \{8,9,10,11\}$ are shown.}
    \label{fig:comparison_sse_runtime_m_qpeak}
\end{figure*}


\begin{table*}[htbp]
\caption{Performance of qc-kmeans on real-world datasets with fixed\text{-}depth QAOA ($p=1$).
The sketch $z_X$ is estimated via Lazy\text{-}QFF with subsample $B=256$, using $D=6$ candidates per cluster.
The peak\text{-}qubit bound is $q_{\text{peak}}=\max\{D,\lceil\log_2 B\rceil+1\}$.
We allocate $10{,}000$ shots to QAOA; QFF uses $1{,}024$ shots per frequency.}

\label{tab:resuls_syn}
\centering
\begin{sc}
\resizebox{\textwidth}{!}{
\begin{tabular}{llll|ccc|ccc|ccc}
\toprule
 \multirow{2}{*}{Dataset}& \multirow{2}{*}{$n$}& \multirow{2}{*}{$d$}&\multirow{2}{*}{Method}&\multicolumn{3}{c|}{$k$=3} & \multicolumn{3}{c|}{$k$=5} & \multicolumn{3}{c}{$k$=10} \\
 &  & & & SSE & $m$& Time (s) & SSE & $m$& Time (s) & SSE & $m$& Time (s) \\
\midrule
Hemi & 1,955 & 7 &  $k$-Means& 9.754E+06& - & - & 5.253E+06& - & - & 2.580E+06& - & - \\
 & & & Classical CKM& 1.698E+07& 84& - & 8.842E+06& 140& - & 3.932E+06& 280& - \\
 & & & qc-kmeans & 1.820E+07
& 84 & 146.39
& 1.600E+07
& 140 & 245.32
& 1.100E+07& 280 & 487.97 \\
 & & & qc-kmeans\_noise & 1.819E+07& 84 & 4,037.28& 1.603E+07& 140 & 5,992.98& 1.095E+07& 280 & 11{,}090.44\\
\midrule
pr2392 & 2,392 & 2 &  $k$-Means& 2.120E+10& - & - & 1.160E+10& - & - & 5.320E+09& - & - \\
 & & & Classical CKM& 2.120E+10& 24& - & 1.265E+10& 40& - & 5.494E+09& 80& - \\
 & & & qc-kmeans & 2.450E+10
& 24 & 50.01
& 1.270E+10
& 40 & 83.63
& 5.680E+09& 80 & 164.59 \\
 & & & qc-kmeans\_noise & 2.446E+10& 24& 1,214.71& 1.266E+10& 40& 1{,}753.70& 5.680E+09
& 80& 3{,}313.80\\
\midrule
AC & 7,195 & 22 &  $k$-Means& 1{,}984.73& - & - & 1{,}513.67& - & - & 1{,}073.29& - & - \\
 & & & Classical CKM& 2{,}465.76& 264& - & 2{,}063.49& 440& - & 1{,}562.16& 880& - \\
 & & & qc-kmeans & 2{,}120.13& 264 & 413.33 & 1{,}609.12& 440 & 735.59 & 1{,}184.29& 880 & 1{,}459.64\\
 & & & qc-kmeans\_noise & 2{,}120.13& 264& 11,265.83& 1609.12& 440& 17{,}091.23& 1{,}184.29& 880& 33{,}356.07\\
\midrule
Rds\_CNT & 10,000 & 4 &  $k$-Means& 1.491E+07& - & - & 5.372E+06& - & - & 1.610E+06& - & - \\
 & & & Classical CKM& 1.492E+07& 48& - & 5.374E+06& 80& - & 1.737E+06& 160& - \\
 & & & qc-kmeans & 4.180E+07
& 48 & 83.83 & 3.660E+07
& 80 & 149.09 & 2.240E+07
& 160 & 296.67 \\
 & & & qc-kmeans\_noise & 4.179E+07& 48& 2,145.78& 3.661E+07& 80& 3{,}263.80& 2.238E+07& 160& 6{,}326.28\\
\midrule
HTRU2 & 17,898 & 8 &  $k$-Means& 7.932E+07& - & - & 4.180E+07& - & - & 1.813E+07& - & - \\
 & & & Classical CKM& 7.932E+07& 96& - & 4.179E+07& 160& - & 2.459E+07& 320& - \\
 & & & qc-kmeans & 1.850E+08
& 96 & 168.31 & 8.950E+07
& 160& 279.99 & 4.330E+07& 320 & 554.00 \\
 & & & qc-kmeans\_noise & 1.854E+08& 96& 3,996.00& 8.953E+07& 160& 6{,}589.69& 4.332E+07& 320& 12{,}340.62\\
\midrule
rds & 50,000 & 3 &  $k$-Means& 447.32 & - & - & 268.01 & - & - & 129.86 & - & - \\
 & & & Classical CKM& 447.32& 36& - & 272.99& 60& - & 145.82& 120& - \\
 & & & qc-kmeans & 557.00& 36 & 71.71 & 307.30 & 60 & 117.51 & 157.84
& 120 & 192.26 \\
 & & & qc-kmeans\_noise & 556.91& 36& 1,703.29& 307.30& 60& 2{,}459.38& 157.84& 120 & 6{,}326.28\\
\midrule
KEGG & 53,413 & 23 &  $k$-Means& 4.901E+08& - & - & 1.884E+08& - & - & 6.051E+07& - & - \\
 & & & Classical CKM& 1.111E+09& 276& - & 1.111E+09& 460& - & 4.627E+08& 920& - \\
 & & & qc-kmeans & 1.470E+09
& 276 & 467.03
& 1.27E+09 & 460 & 776.79 & 4.930E+08& 920 & 1467.12 \\
 & & & qc-kmeans\_noise & 1.473E+09& 276& 11,418.61& 1.270E+09& 460& 17{,}943.50& 4.930E+08& 920& 34{,}500.08\\
\midrule
urbanGB & 360,177 & 2 &  $k$-Means& 4.130E+05& - & - & 2.020E+05& - & - & 8.790E+04& - & - \\
 & & & Classical CKM& 4.166E+05& 24& - & 2.271E+05& 40& - & 9.706E+04& 80& - \\
 & & & qc-kmeans & 4.230E+05
& 24 & 54.23 & 2.050E+05& 40 & 90.29 & 8.880E+04
& 80 & 178.71
\\
 & & & qc-kmeans\_noise & 4.233E+05& 24& 1,157.22& 2.054E+05& 40& 1{,}674.85& 8.881E+04& 80& 3,317.43\\
\midrule
spnet3D & 434,876 & 3 &  $k$-Means& 2.280E+07& - & - & 8.830E+06& - & - & 2.570E+06& - & - \\
 & & & Classical CKM& 2.278E+07& 36& - & 8.826E+06& 60& - & 2.567E+06& 120& - \\
 & & & qc-kmeans & 6.660E+07
& 36 & 76.56
& 5.720E+07
& 60 & 127.44
& 2.580E+07
& 120 & 250.56
\\
 & & & qc-kmeans\_noise & 6.662E+07& 36& 1,727.18& 5.723E+07& 60& 2,469.50& 2.585E+07& 120& 4,818.33\\
\bottomrule
\end{tabular}
}
\end{sc}
\end{table*}

\begin{table*}[htbp]
\caption{Ablation on representative real datasets. We compare \textbf{grouped} (per-group surrogate), \textbf{coupled} (CKM-coupled, sketch-level), and \textbf{exhaustive} (full candidate enumeration, centroid-level). All methods share the same candidate set, one-hot penalty, seeds, and jitter schedule. Reported values are SSE (WCSS) in the original input space 
.}
\label{tab:ablation_result}
\centering
\begin{sc}
\resizebox{\textwidth}{!}{%
\begin{tabular}{lll|ccc|ccc|ccc}
\toprule
 \multirow{2}{*}{Dataset}& \multirow{2}{*}{$n$}& \multirow{2}{*}{$d$}&\multicolumn{3}{c|}{$k{=}3$} & \multicolumn{3}{c|}{$k{=}5$} & \multicolumn{3}{c}{$k{=}10$} \\
  &  &  & grouped & coupled & exhaustive & grouped & coupled & exhaustive & grouped & coupled & exhaustive \\
\midrule
hemi     & 1{,}955  & 7  & 1.826E+7 & 1.836E+7 & 1.836E+7 & 1.522E+7 & 1.512E+7 & 1.512E+7 & 1.067E+7 &\multicolumn{2}{c}{No solution found} \\
pr2392   & 2{,}392  & 2  & 2.459E+10 & 2.656E+10 & 2.656E+10 & 1.288E+10 & 1.335E+10 & 1.364E+10 & 6.116E+9 & \multicolumn{2}{c}{No solution found} \\
AC\_FL   & 7{,}195  & 22 & 2,098.28  & 2,100.33  & 2,100.33  & 1,586.38  & 1,587.42  & 1,587.42  & 1,172.10  & \multicolumn{2}{c}{No solution found}  \\
rds\_cnt & 10{,}000 & 4  & 4.206E+7 & 4.086E+7 & 4.086E+7 & 3.635E+7 & 3.614E+7 & 3.614E+7 & 2.151E+7 & \multicolumn{2}{c}{No solution found}  \\
HTRU2\_L & 17{,}898 & 8  & 1.851E+8 & 1.852E+8 & 1.852E+8 & 8.969E+7 & 8.821E+7 & 8.821E+7 & 4.183E+7 & \multicolumn{2}{c}{No solution found}  \\
KEGG\_FL & 53{,}413 & 23 & 1.461E+9 & 1.466E+9 & 1.466E+9 & 1.265E+9 & 1.265E+9 & 1.265E+9 & 4.865E+8  & \multicolumn{2}{c}{No solution found}  \\
\bottomrule
\end{tabular}}
\end{sc}
\end{table*}



On synthetic datasets, we benchmark qc-kmeans against q-k-Means \cite{poggiali2024hybrid}, q-means \cite{doriguello2023you}, and a hybrid QUBO, evaluating SSE, peak logical qubits ($q_{\mathrm{peak}}$), and wall-clock time (Table~\ref{tab:baselines}). To ensure comparability, all runs omit noise and use a uniform 4-hour budget. We report only converged runs, and we exclude q-means runtime as it assumes QRAM, which is unavailable in our setup. Within this common setup, qc-kmeans consistently attains the lowest SSE and fastest times while keeping $q_{\mathrm{peak}}\le 9$, whereas q-means shows rising SSE as $n$ increases despite using more qubits, and q-k-Means holds $q_{\mathrm{peak}}=5$ but yields poorer SSE and much longer times as $k$ and $n$ grow. Turning from headline results to robustness, varying the largest-instance $q_{\mathrm{peak}}$ and its frequency $m$ has little effect on average SSE within each dataset (Fig.~\ref{fig:comparison_sse_runtime_m_qpeak}).
This suggests that, in the qubit range considered, increasing $q_{\mathrm{peak}}$ does not necessarily lead to consistent improvements in the SSE error metric across all datasets. The average runtime tends to increase with $q_{\mathrm{peak}}$ within each dataset, though the magnitude of this increase varies between datasets.

\subsection{Robustness to Noise} \label{ssec:robustness}

To evaluate scalability and robustness of our algorithm under simulated quantum hardware, we benchmark qc-kmeans on nine real-world datasets using classical $k$-means~\cite{lloyd_least_1982} and compressive $k$-means (CKM)~\cite{Keriven2017Compressive}. Leveraging the fixed-size sketch $z_X$, both qc-kmeans variants can handle datasets with $>400{,}000$ samples and $d$ up to 23, with a maximum per-dataset runtime of 8 minutes in the ideal-hardware setting. Qc-kmeans approaches the solution quality of classical baselines on several datasets, including \texttt{PR2392}, \texttt{KEGG}, \texttt{RDS}, and \texttt{URBANGB}.
We also assess noise sensitivity of qc-kmean using the \texttt{FakeMelbourneV2} backend, which emulates the behavior of an IBM Quantum 14-qubit device. The wall-clock budget was increased to 12 hours to account for the additional cost of transpilation, sampling (shots), and noisy objective evaluations under calibrated noise models. Across all datasets and $k$ settings, Table~\ref{tab:resuls_syn} shows small SSE differences between ideal and noisy simulations, indicating low sensitivity to realistic noise; in a few cases, noise even yields slight SSE improvements. These results support the viability of qc-kmeans on NISQ devices: within a fixed resource budget and constant-qubit, shallow-depth regime, the method consistently converges and shows low drift under noise. 

\subsection{Ablation Studies}\label{ssec:ablation}
We evaluate the effect of the per-group decomposition by comparing the grouped formulation with the coupled and exhaustive references. Table~\ref{tab:ablation_result} reports results in terms of SSE (WCSS).
On six representative datasets with $k\in\{3,5\}$, the gap between grouped and coupled SSE remains consistently small. 
The exhaustive results closely follow the coupled ones on datasets such as AC\_FL and KEGG\_FL, confirming that both capture the same surrogate structure, with minor discrepancies when SSE is measured in the original space. For $k=10$, only grouped results are available since the other variants become intractable due to the combinatorial growth in $\prod_g D_g$.
These findings indicate that the per-group relaxation maintains the fidelity of the clustering objective while enabling a decomposition suitable for NISQ execution.

\section{Discussion}
\paragraph{Limitations} \label{ssec:limitations}
On current NISQ hardware, coherence times and gate errors constrain practical QAOA depth to very small values, which can degrade solution quality on more complex or tightly constrained instances. Our evaluation includes a noise simulation, but a systematic assessment on real hardware remains open. Additionally, we acknowledge that our hybrid method can inherit known $k$-means issues. These limitations motivate future work on error mitigation and parameter transfer for deeper $p$, hardware runs, and more robust initialization/mixing strategies.

\paragraph{Baseline Selection} \label{ssec:baseline_selection} 
For q-means and q-k-Means, we report SSE and qubits and omit runtime when timing is not directly comparable. At the target sizes, scalable public implementations of these quantum baselines are not available due to NISQ-era constraints (circuit width/depth and sampling budgets) and implementation limits, so several variants do not compile or converge within the budget. Our hybrid design does not provide theoretical or empirical speedups versus classical k-means/CKM, and we expect classical methods to remain preferable when quantum resources are unnecessary. The value here is constant qubit width and shallow depth aligned with NISQ feasibility. 

\color{black}
\section{Conclusion} 
\label{sec:conclusion}
We introduced qc-kmeans, a NISQ-friendly hybrid approach that compresses a dataset into a fixed-size Fourier-feature sketch and delegates the discrete centroid selection to a shallow QAOA with an XY mixer under one-hot constraints. 
This formulation decouples quantum register width from the number of samples and keeps the quantum subproblem independent of $N$, while seeding and reassignment remain $N$-dependent. Across synthetic and real datasets, the method attains competitive clustering quality while using constant-width, shallow circuits and shows robustness under noise simulation compared to other quantum baselines. 

\color{black}
\bibliographystyle{IEEEtran}
\bibliography{icml_paper}

@article{doriguello2023you,
  title={Do you know what q-means?},
  author={Doriguello, Jo{\~a}o F and Luongo, Alessandro and Tang, Ewin},
  journal={arXiv preprint arXiv:2308.09701},
  year={2023}
}

@article{preskill2018quantum,
  title={Quantum computing in the NISQ era and beyond},
  author={Preskill, John},
  journal={Quantum},
  volume={2},
  pages={79},
  year={2018},
  publisher={Verein zur F{\"o}rderung des Open Access Publizierens in den Quantenwissenschaften}
}

@Article{wang_predicting_2022,
  author = {Wang, Edward and Ballachay, Riley and Cai, Genpei and Cao, Yankai and Trajano, Heather L.},
  title = {Predicting xylose yield from prehydrolysis of hardwoods: A machine learning approach},
  journal = {Frontiers in Chemical Engineering},
  year = {2022},
  volume = {4},
}

@Article{hartigan_algorithm_1979,
  author  = {Hartigan, J. A. and Wong, M. A.},
  title   = {Algorithm AS 136: A K-Means Clustering Algorithm},
  journal = {Journal of the Royal Statistical Society. Series C (Applied Statistics)},
  year    = {1979},
  volume  = {28},
  number  = {1},
  pages   = {100--108},
}

@Article{rao_cluster_1971,
  author  = {Rao, M. R.},
  title   = {Cluster Analysis and Mathematical Programming},
  journal = {Journal of the American Statistical Association},
  year    = {1971},
  volume  = {66},
  number  = {335},
  pages   = {622--626},
}

@Article{padberg_branch-and-cut_1991,
  author  = {Padberg, Manfred and Rinaldi, Giovanni},
  title   = {A Branch-and-Cut Algorithm for the Resolution of Large-Scale Symmetric Traveling Salesman Problems},
  journal = {SIAM Review},
  year    = {1991},
  volume  = {33},
  number  = {1},
  pages   = {60--100},
}

@Article{lloyd_least_1982,
  author  = {Lloyd, S.},
  title   = {Least squares quantization in PCM},
  journal = {IEEE Transactions on Information Theory},
  year    = {1982},
  volume  = {28},
  number  = {2},
  pages   = {129--137},
}

@Article{jain_data_2010,
  author  = {Jain, Anil K.},
  title   = {Data clustering: 50 years beyond K-means},
  journal = {Pattern Recognition Letters},
  year    = {2010},
  volume  = {31},
  number  = {8},
  pages   = {651--666},
}

@misc{Dua:2019,
  author = {Dua, Dheeru and Graff, Casey},
  year = {2017},
  title = {UCI Machine Learning Repository},
  url = {http://archive.ics.uci.edu/ml},
  institution = {University of California, Irvine, School of Information and Computer Sciences}
}

@inproceedings{saiphet2021quantum,
  title={Quantum approximate optimization and k-means algorithms for data clustering},
  author={Saiphet, Jirawat and Suwanna, Sujin and Chotibut, Thiparat and Chantasri, Areeya},
  booktitle={Journal of Physics: Conference Series},
  volume={1719},
  number={1},
  pages={012100},
  year={2021},
  organization={IOP Publishing}
}

@article{chen2025provably,
  title={Provably faster randomized and quantum algorithms for $ k $-means clustering via uniform sampling},
  author={Chen, Tyler and Ray, Archan and Seshadri, Akshay and Herman, Dylan and Bach, Bao and Deshpande, Pranav and Som, Abhishek and Kumar, Niraj and Pistoia, Marco},
  journal={arXiv preprint arXiv:2504.20982},
  year={2025}
}

@article{poggiali2024hybrid,
   title={Quantum clustering with k-Means: A hybrid approach},
   volume={992},
   ISSN={0304-3975},
   url={http://dx.doi.org/10.1016/j.tcs.2024.114466},
   DOI={10.1016/j.tcs.2024.114466},
   journal={Theoretical Computer Science},
   publisher={Elsevier BV},
   author={Poggiali, Alessandro and Berti, Alessandro and Bernasconi, Anna and Del Corso, Gianna M. and Guidotti, Riccardo},
   year={2024},
   month=apr, pages={114466} }

@misc{tomesh_coreset_2020,
	title = {Coreset {Clustering} on {Small} {Quantum} {Computers}},
	url = {http://arxiv.org/abs/2004.14970},
	doi = {10.48550/arXiv.2004.14970},
	urldate = {2025-08-18},
	publisher = {arXiv},
	author = {Tomesh, Teague and Gokhale, Pranav and Anschuetz, Eric R. and Chong, Frederic T.},
	month = apr,
	year = {2020},
	note = {arXiv:2004.14970 [quant-ph]},
}

@article{hao2024end,
  title={End-to-end protocol for high-quality QAOA parameters with few shots},
  author={Hao, Tianyi and He, Zichang and Shaydulin, Ruslan and Larson, Jeffrey and Pistoia, Marco},
  journal={arXiv preprint arXiv:2408.00557},
  year={2024}
}

@article{kerenidis2019q,
  title={q-means: A quantum algorithm for unsupervised machine learning},
  author={Kerenidis, Iordanis and Landman, Jonas and Luongo, Alessandro and Prakash, Anupam},
  journal={Advances in neural information processing systems},
  volume={32},
  year={2019}
}

@misc{yogendran2024big,
      title={Big data applications on small quantum computers}, 
      author={Boniface Yogendran and Daniel Charlton and Miriam Beddig and Ioannis Kolotouros and Petros Wallden},
      year={2024},
      eprint={2402.01529},
      archivePrefix={arXiv},
      primaryClass={quant-ph},
      url={https://arxiv.org/abs/2402.01529}, 
}

@ARTICLE{DiAdamo2022Practical,
  author={DiAdamo, Stephen and O’Meara, Corey and Cortiana, Giorgio and Bernabé-Moreno, Juan},
  journal={IEEE Transactions on Quantum Engineering}, 
  title={Practical Quantum K-Means Clustering: Performance Analysis and Applications in Energy Grid Classification}, 
  year={2022},
  volume={3},
  number={},
  pages={1-16},
  doi={10.1109/TQE.2022.3185505}}

@INPROCEEDINGS{Keriven2017Compressive,
  author={Keriven, Nicolas and Tremblay, Nicolas and Traonmilin, Yann and Gribonval, Rémi},
  booktitle={2017 IEEE International Conference on Acoustics, Speech and Signal Processing (ICASSP)}, 
  title={Compressive K-means}, 
  year={2017},
  volume={},
  number={},
  pages={6369-6373},
  doi={10.1109/ICASSP.2017.7953382}}

@InProceedings{xue2023Near,
  title = 	 {Near-Optimal Quantum Coreset Construction Algorithms for Clustering},
  author =       {Xue, Yecheng and Chen, Xiaoyu and Li, Tongyang and Jiang, Shaofeng H.-C.},
  booktitle = 	 {Proceedings of the 40th International Conference on Machine Learning},
  pages = 	 {38881--38912},
  year = 	 {2023},
  editor = 	 {Krause, Andreas and Brunskill, Emma and Cho, Kyunghyun and Engelhardt, Barbara and Sabato, Sivan and Scarlett, Jonathan},
  volume = 	 {202},
  series = 	 {Proceedings of Machine Learning Research},
  month = 	 {23--29 Jul},
  publisher =    {PMLR},
}

@misc{harrow_small_2020,
	title = {Small quantum computers and large classical data sets},
	url = {http://arxiv.org/abs/2004.00026},
	doi = {10.48550/arXiv.2004.00026},
	urldate = {2025-08-18},
	publisher = {arXiv},
	author = {Harrow, Aram W.},
	month = mar,
	year = {2020},
	note = {arXiv:2004.00026 [quant-ph]},
}

@article{fuchs2021efficient,
  title={Efficient encoding of the weighted max k-cut on a quantum computer using qaoa},
  author={Fuchs, Franz G and Kolden, Herman {\O}ie and Aase, Niels Henrik and Sartor, Giorgio},
  journal={SN Computer Science},
  volume={2},
  number={2},
  pages={89},
  year={2021},
  publisher={Springer}
}

@article{brassard2000quantum,
  title={Quantum amplitude amplification and estimation},
  author={Brassard, Gilles and Hoyer, Peter and Mosca, Michele and Tapp, Alain},
  journal={arXiv preprint quant-ph/0005055},
  year={2000}
}

@article{lloyd2014qpca,
  title={Quantum principal component analysis},
  author={Lloyd, Seth and Mohseni, Masoud and Rebentrost, Patrick},
  journal={Nature physics},
  volume={10},
  number={9},
  pages={631--633},
  year={2014},
  publisher={Nature Publishing Group UK London}
}

@article{kerenidis2021qspectral,
  title={Quantum spectral clustering},
  author={Kerenidis, Iordanis and Landman, Jonas},
  journal={Physical Review A},
  volume={103},
  number={4},
  pages={042415},
  year={2021},
  publisher={APS}
}

@article{giovannetti2008qram,
  title={Quantum random access memory},
  author={Giovannetti, Vittorio and Lloyd, Seth and Maccone, Lorenzo},
  journal={Physical review letters},
  volume={100},
  number={16},
  pages={160501},
  year={2008},
  publisher={APS}
}

@article{giurgica2022lowdepth,
  title={Low depth algorithms for quantum amplitude estimation},
  author={Giurgica-Tiron, Tudor and Kerenidis, Iordanis and Labib, Farrokh and Prakash, Anupam and Zeng, William},
  journal={Quantum},
  volume={6},
  pages={745},
  year={2022},
  publisher={Verein zur F{\"o}rderung des Open Access Publizierens in den Quantenwissenschaften}
}

@article{suzuki2020qae,
  title={Amplitude estimation without phase estimation: Y. Suzuki et al.},
  author={Suzuki, Yohichi and Uno, Shumpei and Raymond, Rudy and Tanaka, Tomoki and Onodera, Tamiya and Yamamoto, Naoki},
  journal={Quantum Information Processing},
  volume={19},
  number={2},
  pages={75},
  year={2020},
  publisher={Springer}
}

@article{weiss2024quantum,
  title={Quantum random access memory architectures using 3D superconducting cavities},
  author={Weiss, DK and Puri, Shruti and Girvin, SM},
  journal={PRX Quantum},
  volume={5},
  number={2},
  pages={020312},
  year={2024},
  publisher={APS}
}

\end{document}